\pdfoutput=1
\RequirePackage{ifpdf}
\ifpdf % We~are running pdfTeX in pdf mode
\documentclass[pdftex]{sigma}
\else
\documentclass{sigma}
\fi

\usepackage{enumerate}
\usepackage{empheq}

\newtheorem{thm}{Theorem}[section]
\newtheorem{cor}[thm]{Corollary}
\newtheorem{lem}[thm]{Lemma}
\newtheorem{prop}[thm]{Proposition}

\theoremstyle{definition}
\newtheorem{defn}[thm]{Definition}
\newtheorem{rem}[thm]{Remark}

\numberwithin{equation}{section}

\newcommand{\norm}[1]{\left\Vert#1\right\Vert}
\renewcommand{\sc}[2]{\langle #1|#2 \rangle}
\newcommand{\pair}[2]{\left\langle #1 , #2 \right\rangle}

\newcommand{\proj}[2]{|#1 \rangle \langle #2|}

\newcommand{\R}{\mathbb R}

\newcommand{\C}{\mathbb C}
\newcommand{\N}{\mathbb N}
\renewcommand{\H}{\mathcal{H}}
\newcommand{\g}{\mathfrak{g}}

\DeclareMathOperator{\Gr}{Gr_{res}(\H,\H_+)}
\newcommand{\U}{\mathrm{U}(\H)}
\newcommand{\Up}{\mathrm{U}(\H_+)}\newcommand{\Um}{\mathrm{U}(\H_-)}
\newcommand{\Ur}{\mathrm{U}_{\textnormal{res}}(\H)}
\newcommand{\Urt}{\widetilde{\mathrm{U}}_{\textnormal{res}}(\H)}
\newcommand{\UU}{\mathcal U(\H)}

\newcommand{\uup}{\mathfrak u(\H_+)}\newcommand{\uum}{\mathfrak u(\H_-)}

\newcommand{\uujp}{\mathfrak u^1(\H_+)}\newcommand{\uujm}{\mathfrak u^1(\H_-)}

\newcommand{\ur}{\mathfrak u_{\textnormal{res}}(\H)}
\newcommand{\urp}{\mathfrak u_{\textnormal{res}}^1(\H)}
\newcommand{\urt}{\widetilde{\mathfrak u}_{\textnormal{res}}(\H)}
\newcommand{\urtp}{\widetilde{\mathfrak u}_{\textnormal{res}}^1(\H)}

\newcommand{\glrp}{\mathfrak{gl}^1_\textnormal{res}(\H)}

\DeclareMathOperator{\Tr}{Tr}
\DeclareMathOperator{\Trr}{Tr_{res}}
\DeclareMathOperator{\ad}{ad}
\DeclareMathOperator{\Ad}{Ad}
\DeclareMathOperator{\const}{const}
\DeclareMathOperator{\im}{im}

\newcommand{\Li}{L^\infty(\H)}
\newcommand{\Lj}{L^1(\H)}
\newcommand{\Ljm}{L^1(\H_-)}\newcommand{\Ljp}{L^1(\H_+)}
\newcommand{\Ld}{L^2(\H)}
\newcommand{\Ldpm}{L^2(\H_+,\H_-)}
\newcommand{\Lp}{L^p(\H)}

\newcommand{\abs}[1]{\left\vert#1\right\vert}
\newcommand{\set}[1]{\left\{#1\right\}}
\newcommand{\pb}{\{\cdot,\cdot\}}

\renewcommand{\to}{\rightarrow}
\newcommand{\tto}{\longrightarrow}

\newcommand{\be}{\begin{equation}}
\newcommand{\ee}{\end{equation}}
\newcommand{\bse}{\begin{subequations}}
\newcommand{\ese}{\end{subequations}}
\newcommand{\ben}{\begin{enumerate}}
\newcommand{\een}{\end{enumerate}}

\begin{document}
\allowdisplaybreaks

\newcommand{\arXivNumber}{2407.21605}

\renewcommand{\PaperNumber}{104}

\FirstPageHeading

\ShortArticleName{Geometry of Integrable Systems Related to the Restricted Grassmannian}

\ArticleName{Geometry of Integrable Systems Related\\ to the Restricted Grassmannian}

\Author{Tomasz GOLI\'NSKI~$^{\rm a}$ and Alice Barbora TUMPACH~$^{\rm b}$}

\AuthorNameForHeading{T.~Goli\'nski and A.B.~Tumpach}

\Address{$^{\rm a)}$~University of Bia{\l}ystok, Cio{\l}kowskiego 1M, 15-245 Bia{\l}ystok, Poland}
\EmailD{\href{mailto:tomaszg@math.uwb.edu.pl}{tomaszg@math.uwb.edu.pl}}
\URLaddressD{\url{https://tomaszg.pl/math/}}

\Address{$^{\rm b)}$~Institut CNRS Pauli, UMI CNRS 2842, Oskar-Morgenstern-Platz 1, 1090 Wien, Austria}
\EmailD{\href{mailto:barbara.tumpach@math.cnrs.fr}{barbara.tumpach@math.cnrs.fr}}
\URLaddressD{\url{https://geometricgreenlearning.com/}}

\ArticleDates{Received August 01, 2024, in final form November 12, 2024; Published online November 22, 2024}

\Abstract{A hierarchy of differential equations on a Banach Lie--Poisson space related to the restricted Grassmannian is studied. Flows on the groupoid of partial isometries and on the restricted Grassmannian are described, and a momentum map picture is presented.}

\Keywords{integrable systems; momentum map; Banach Lie--Poisson spaces; partial isometries; restricted Grassmannian; Magri method}

\Classification{70H06; 53D20; 53D17; 46T05; 47B10}

\section{Introduction}

In this paper, we are dealing with integrable systems, introduced first in \cite{GO-grass}, on a certain Banach Lie--Poisson space related to the restricted Grassmannian $\Gr$. Let us clarify that by the term ``integrable system'' in infinite dimensions we mean a system (or hierarchy) of differential equations possessing an infinite family of independent commuting integrals of motion, see, e.g., introduction in \cite{hitchin-segal-ward} or \cite{bolsinov-na4}. We do not imply complete integrability and do not claim to be able to write down general solutions.
However, in particular, cases we are able to linearize these systems and obtain explicit formulas for solutions.

The restricted Grassmannian is an infinite-dimensional K\"ahler manifold modelled on a Hilbert space that first appeared in the theory of the Korteweg--de Vries equation, see \cite{sato-sato,segal-wilson}, and is also related to fermionic second quantization of quantum field theory \cite{chiumiento24,mickelsson,wurzbacher} and loop groups \cite{segal, segal-wilson}. Using the theory of Banach Lie--Poisson spaces \cite{OR, Oext}, it was shown in~\cite{Ratiu-grass} that the restricted Grassmannian is an affine co-adjoint orbit of the Banach Lie group~$\Ur$ consisting of unitary operators on a polarized Hilbert space $\mathcal{H}$ with off-diagonal blocks Hilbert--Schmidt. It turns out that $\Ur$ is a Banach Poisson--Lie group in a non-trivial way \cite{tumpach-bruhat} and that the restricted Grassmannian inherits from $\Ur$ a non-trivial Poisson structure whose symplectic leaves are the Bruhat--Schubert cells.
Moreover, the (co-)tangent space of the restricted Grassmannian is naturally endowed with a structure of strong infinite-dimensional hyperk\"ahler manifold \cite{tumpach-hyperkahler}, which can be identified with a complex affine co-adjoint orbit of the complexification of $\Ur$ \cite{tumpach-hyperkaehler2}. The diffeomorphism between (co-)tangent space and complex orbit exists for all Hermitian-symmetric affine co-adjoint orbits of compact type \cite{tumpach-hyperkaehler2,tumpach-classification}. For Hermitian-symmetric affine co-adjoint orbits of non-compact type, like the restricted Siegel disc~\cite{T-siegel}, the existence of analogous hyperk\"ahler structures is investigated in~\cite{GBRT-hyperkaehler}.

Let us recall that affine co-adjoint orbits of Banach--Lie groups are co-adjoint orbits of their central extensions. In particular, the Lie algebra $\ur$ of the group $\Ur$ possesses a central extension $\urt$ given by the cocycle called Schwinger term, which describes the failure of the fermionic second quantization procedure of being a Lie algebra homomorphism \cite{mickelsson,schwinger,wurzbacher}. From our point of view, it is interesting that this central extension leads to a bi-Hamiltonian structure (i.e., a pencil of compatible Poisson brackets) on the predual Banach space $\urp$ of the Banach Lie algebra $\ur$. Namely, the Banach Lie--Poisson bracket of the central extension of $\urp$ as introduced in \cite{Ratiu-grass} leads to the pencil of Poisson brackets on $\urp$ if we reinterpret the variable of central extension as a pencil parameter. In this approach, we can discover a~version of a ``frozen bracket'' by Mishchenko and Fomenko \cite{fomenko-trofimov}, but with differences due to the subtleties of the infinite-dimensional case. In the paper \cite{GO-grass}, a~family of Casimirs for this pencil was studied. Using Magri method, an infinite family of integrals of motion in involution for a~hierarchy of Hamilton equations is constructed.

Many properties of this hierarchy of equations were studied in the papers \cite{GO-grass-bial,GO-grass,GO-grass2} and more recently in \cite{GT-partiso-bial}. These papers also include some examples, including finite-dimensional ones. Moreover, the paper \cite{GO-grass2} describes possible physical interpretation of these integrable systems to the description of the interaction of electromagnetic waves with a nonlinear dielectric medium, including Kerr effect and parametric conversion. Finding more applications of the results is still an open problem. The aim of the present paper is to present a more geometrical picture of some of the results mentioned above, and to study the flows induced by those equations on the set of partial isometries.

Note here that the set of partial isometries of a given Hilbert space (or even in an abstract~$C^*$ or $W^*$-algebra) possesses a structure of Banach Lie groupoid, see \cite{BGJP,OS}. An analogous structure can also be introduced for the set of partial isometries associated to the restricted Grassmannian, see~\cite{GJS-partiso}.

The paper is organized as follows. Section \ref{sec:gr} is devoted to preliminary notions about Banach Lie--Poisson spaces, the restricted Grassmannian and to the construction of a pencil of Poisson brackets. Section \ref{sec:int} recalls the construction of a hierarchy of integrable systems from \cite{GO-grass} adapting some formulas to the particular case discussed in the present paper.
Section \ref{sec:momentum} describes the momentum map for the hierarchy and comments on the structure of the symplectic leaves.
Section \ref{sec:partiso} deals with the particular case of the equations under an assumption that one diagonal block vanishes. Under this condition the system descends to the Banach Lie groupoid of partial isometries. The partial isometries evolve in a particular way which preserves the initial space. An example solution in the case of partial isometries of rank one is also presented.
Section \ref{sec:gr-lin} is dedicated to studying the considered system when restricted to a family of particular coadjoint orbits for the coadjoint action of the central extension $\Urt$ of the Banach Lie group $\Ur$, which are diffeomorphic (even symplectomorphic) to the restricted Grassmannian $\Gr$. It was shown in \cite{GO-grass} that using homogeneous coordinates on $\Gr$ the systems become linear. In this section, a more general result is obtained --- namely that the systems become linear on the coadjoint orbit itself. A solution for some equations from the hierarchy on $\Gr$ is presented.

\section[Restricted Grassmannian, related Banach Lie algebras, and Poisson brackets]{Restricted Grassmannian, related Banach Lie algebras,\\ and Poisson brackets}\label{sec:gr}

In this section, we will recall necessary information about the restricted Grassmannian and objects around it, which will be used in the paper. More detailed exposition and background can be found, e.g., in \cite{Ratiu-grass,GO-grass,mickelsson,OR,segal,wurzbacher}.

A fundamental notion needed in the present paper is the notion of Banach Lie--Poisson space, which was introduced in the paper \cite[Definition 4.1 and Theorem 4.2]{OR}. It was useful, e.g., in the study of such integrable systems as the infinite Toda lattice \cite{Oind}. Later the notion was also extended to arbitrary duality pairing in \cite[Definition 3.12 and Theorem 3.14]{tumpach-bruhat}, but this extension is not used in the present paper.

\begin{defn}\label{def:blp}
A Banach Lie--Poisson space is a Banach space $\g_*$ predual to a Banach Lie algebra $\g$ such that $\g_*\subset \g^*$ is preserved by the coadjoint action $\ad^*\colon\g^*\to\g^*$,
$\ad^*_\g \g_*\subset \g_*$
together with the canonical structure of Banach Poisson manifold given by the bracket
\[ %\label{blp-pb}
\{f,g\}(\mu) = \langle \mu, [Df(\mu),Dg(\mu)] \rangle \]
for $f,g\in C^\infty(\g_*)$.
\end{defn}

In the formula above, we treat the derivatives $Df(\mu)$ and $Dg(\mu)$ as elements of the Banach Lie algebra $(\g_*)^*=\g$. The Hamiltonian vector field for a Hamiltonian $h\in C^\infty(\g_*)$ with respect to this bracket assumes the form
\be \label{blp-hvf}
X_h(\mu) = -\ad^*_{Dh(\mu)}\mu.\ee

Now, in order to define the restricted Grassmannian, we consider a complex separable Hilbert space $\H$ with a fixed polarization, i.e., orthogonal decomposition
\be\label{polarization}\H=\H_+\oplus \H_-.\ee
We will denote by $P_+$ and $P_-$ the orthogonal projectors onto $\H_+$ and $\H_-$ respectively. It is usually assumed that both closed subspaces $\H_\pm$ are infinite-dimensional.
Finite-dimensional case can also be studied. However, in that case some geometrical properties break down, e.g., $\Gr$ as defined below is no longer a homogeneous space (in particular, it contains Grassmannians of subspaces of different fixed finite dimensions). An example of solution of the equations considered in the present paper in the finite-dimensional case was presented in \cite{GT-partiso-bial}.

Given the polarization \eqref{polarization}, one introduces a block decomposition of an operator $A$ acting on $\H$
\be \label{blocks}
A=\left(
\begin{matrix}
 A_{++} & A_{+-}\\
 A_{-+} & A_{--}
\end{matrix}
\right).\ee
To simplify the notation, we will sometimes identify the operators $A_{\pm\pm}\colon\H_\pm\to\H_\pm$ with $P_\pm A P_\pm\colon\H\to\H$ when it will lead to no confusion.

Let $\Lp$ denote the Schatten class of operators acting on $\H$ equipped with the norm
\[
 \norm{A}_p = \bigl(\Tr \abs A^p\bigr)^{1/p}.
 \]
The $\Lp$ spaces are (not closed) ideals in the $C^*$-algebra $\Li$ of bounded operators in~$\H$. Note that norm-closure of any space $\Lp$ is the ideal of compact operators, which is the only closed ideal in $\Li$. In particular, $\Lj$ denotes the ideal of trace-class operators and $\Ld$ is the ideal of Hilbert--Schmidt operators. An important property is that the dual spaces of these ideals are $\Li$ and $\Ld$ respectively and the pairing is given by the trace.

\begin{defn}
The \emph{restricted Grassmannian} $\Gr$ is the set of closed subspaces ${W\subset\H}$ such that
\ben[(i)]\itemsep=0pt
\item the orthogonal projection $p_+\colon W\to \H_+$ is a Fredholm operator;
\item the orthogonal projection $p_-\colon W\to \H_-$ is a Hilbert--Schmidt operator.
\een
\end{defn}

It is sometimes useful to work with orthogonal projection onto the element $W\in \operatorname{Gr}_{\rm res}(\H, \H_+)$, which we will denote by $P_W$. In this way, one can identify the set $\Gr$ with the set of orthogonal projections $\{ P_W \mid W\in\Gr\}$ in $\H$.

In the paper \cite{spera-valli}, it was demonstrated that one can easily describe the projections coming from $\Gr$.
\begin{prop}\label{prop:gr-proj}
\[ W\in\Gr \ \Longleftrightarrow \ P_W-P_+ \in \Ld. \]
\end{prop}
One also introduces the Banach Lie group $\Ur$ acting transitively on $\Gr$,
\[
\Ur:=\bigl\{ u\in \U \mid [u,P_+]\in \Ld\bigr\},
\]
where $\U$ is the full unitary group of operators acting on $\H$. The Banach Lie algebra of~$\Ur$~is
\[
\ur:=\bigl\{ X\in \Li \mid X^* = -X, [X,P_+]\in \Ld\bigr\}.\]
Note that the commutator condition in these definitions means that off-diagonal blocks of the operators with respect to the block decomposition \eqref{blocks} are Hilbert--Schmidt operators. Instead of considering the commutator with $P_+$, one can equivalently write this condition as $[u, \epsilon] = 0$ with an involution $\epsilon = P_+-P_- = 2P_+-1$. In the case of the group $\Ur$, invertibility implies that diagonal blocks are Fredholm.

The restricted Grassmannian $\Gr$ can be viewed as the homogeneous space $\Ur/\allowbreak(\Up\times\Um)$, where $\Up\times\Um$ is the stabilizer of $\H_+$ with respect to the natural action.
Moreover, the group $\Ur$ possesses a structure of Banach Poisson--Lie group, see~\cite[Theorem 7.12]{tumpach-bruhat} and the subgroup $\Up\times\Um$ is a Banach Poisson--Lie
subgroup \cite[Proposition 8.2]{tumpach-bruhat}. Consequently, the restricted Grassmannian is not only a symplectic manifold (as affine coadjoint orbit, see~\cite{Ratiu-grass}) but admits also a non-trivial Bruhat Poisson structure, see~\cite[Theorem 8.3]{tumpach-bruhat}.

In the papers \cite{Ratiu-grass,Oext}, a structure of Banach Lie--Poisson space on the predual
\[
 \urp:=\bigl\{ \mu\in \ur \mid \mu_{++}\in \Ljp , \mu_{--}\in \Ljm\bigr\}
 \]
of $\ur$ was investigated, see Definition \ref{def:blp}.
The duality pairing between $\mu\in\urp$ and $A\in\ur$ is given by
\be \label{pair-urp}\pair\mu A:=\Trr(\mu A),\ee
where $\Trr$ is the \emph{restricted trace} defined on $\urp$ by
$\Trr \mu :=\Tr(\mu_{++}+\mu_{--})$.
Note that~$\Trr$ is defined on a larger domain than $\Lj$ and it coincides with the standard trace $\Tr$ on trace-class operators.
The properties of the restricted trace are similar to the properties of the standard trace but one needs to be more careful. For instance, one has
\be \label{trr-cyclic}\Trr(\mu\nu)=\Trr(\nu\mu)\ee
for $\mu\in\urp$ and $\nu\in\ur$,
see \cite{GO-grass} for details.

From the pairing \eqref{pair-urp}, we conclude that $\bigl(\urp\bigr)^*\cong \ur$, i.e., the Banach space
$\urp$ is predual to $\ur$.
It is straightforward that $\urp$ is preserved by the coadjoint action due to the fact that $L^p$ spaces are ideals. Thus, $\urp$ is a Banach Lie--Poisson space with the Poisson bracket
\be \label{pb-0}\{f,g\}_0(\mu) =
\Trr\bigl(\mu [Df(\mu),Dg(\mu)]\bigr).\ee

Next step is to construct a central extension of $\ur$ by a cocycle called the Schwinger term (see \cite{Ratiu-grass, GO-grass, mickelsson,schwinger, wurzbacher})
\be \label{schwinger}s(X,Y)=\Tr(X_{+-}Y_{-+}-Y_{+-}X_{-+}),\ee
where $X, Y \in \ur$.
The cocycle $s$ gives rise to a Banach Lie algebra structure on
\[\urt:=\ur\oplus {\rm i}\R\]
with the following Lie bracket
$ [(X,\gamma),(Y,\gamma')]=\bigl([X,Y],-s(X,Y)\bigr)$.

Naturally, the predual of $\urt$ is
$\urtp:=\urp\oplus {\rm i}\R$
with the pairing given by
$ \pair{(\mu,\gamma)}{(X,\gamma)}_{\sim}=\Trr(\mu X)+\gamma\lambda$,
for $\mu\in\urp$, $X\in\ur$, $\gamma,\lambda\in {\rm i}\R$.
In consequence, the Banach space $\urtp$ also possesses a structure of Banach Lie--Poisson space and the Poisson bracket is given by
\begin{align*}
 \{F,G\}(\mu,\gamma)& = \pair{(\mu,\gamma)}{[DF(\mu,\gamma),DG(\mu,\gamma)]}_\sim\\
& =\Trr\bigl(\mu [D_1F(\mu,\gamma),D_1G(\mu,\gamma)]\bigr)-\gamma s(D_1F(\mu,\gamma),D_1G(\mu,\gamma)),
\end{align*}
where $D_1$ is the derivation with respect to the first argument of functions $F,G\in C^\infty\bigl(\urtp\bigr)$.
Note that since the extension is central, there is no derivative with respect to $\gamma$ in this Poisson bracket. Thus, we can
consider the variable $\gamma$ as a parameter and obtain a pencil of Poisson brackets on $\urp$
\be \label{pb-pencil}\{f,g\}_\gamma(\mu) =
\{f,g\}_0(\mu) - \gamma \{f,g\}_s(\mu) \ee
for $f,g\in C^\infty\bigl(\urp\bigr)$, where $\pb_0$ is the Lie--Poisson bracket \eqref{pb-0} of $\urp$ and $\pb_s$ is the Schwinger bracket
\be \label{pb-s}\{f,g\}_s(\mu) = s(Df(\mu),Dg(\mu))= \Tr (Df(\mu)_{+-}Dg(\mu)_{-+} - Dg(\mu)_{+-}Df(\mu)_{-+}).\ee

\begin{rem}\label{frozen}
While working with the restricted trace one needs to keep in mind that in general it is not possible to perform cyclic permutations in expressions containing commutators. As an example, consider the Schwinger term \eqref{schwinger} which can be expressed as
\begin{equation}\label{schwinger_expression}
s(X, Y) = \Trr\bigl(X[Y,P_+]\bigr),
\end{equation}
where $X$ and $Y$ belong to $\ur$. The result of a cyclic permutation would be, for example, $\Trr(P_+[X, Y])$, which does not make sense in general.

If one disregards that for a second, one can write a formal expression
\[ \{f,g\}_s(\mu)=-\Trr\bigl( P_+ [DF(\mu),DG(\mu)]\bigr). \]
It looks just like a ``frozen bracket'' discovered by Mishchenko and Fomenko, but the freezing point $P_+$ does not lie in the Banach Lie--Poisson space. See, e.g., \cite{bols-bor,fomenko-trofimov} for the review of this approach in the finite-dimensional setting. Generally when one considers a Banach Lie--Poisson bracket on the predual $\g_*$ of a Banach Lie algebra, one can consider the freezing point as any element of $\g_*$. However, the formula makes sense also for an element of $\g^*$, which in infinite-dimensional case can be larger. Additionally, the closed subspace spanned by the commutators of elements in $\g$ can be significantly smaller than the whole $\g$ even taking into account zero-trace condition, see, e.g., \cite{kalton89} and references therein. Regretfully in the case we consider, $P_+$ does not belong even to the dual of the subspace generated by commutators of operators from $\ur$. Even though, this insight allows one to correctly predict the form of the Casimir functions.
\end{rem}

There exists a central $U(1)$-extension $\Urt$ of the group $\Ur$ for which $\urt$ is the Lie algebra. Note that there is no continuous global cross-section $\Urt\to\Ur$. For the details of the construction of $\Urt$, we refer the reader to \cite{mickelsson,segal,wurzbacher}. For the purpose of this paper it is enough to know the form of the coadjoint action of $\Urt$ on $\urtp$, which can be found in \cite{Ratiu-grass,T-siegel,GO-grass}
\be\label{coad-ext}
\Ad^*_{\Gamma}(\mu,\gamma) = \bigl(g^{-1}\mu g + \gamma \bigl(P_+ - g^{-1}P_+ g\bigr), \gamma\bigr),
\ee where $\Gamma\in\Urt$ projects down to $g \in \Ur$.

\section[Hierarchy of integrable systems on u\*1\_res(H)]{Hierarchy of integrable systems on $\boldsymbol{\urp}$}\label{sec:int}

In order to introduce the integrable systems, which are the subject of this paper, we are looking for Casimirs of the pencil \eqref{pb-pencil}, i.e., a family of functions $I_\gamma^n$ on $\urp$ such that
\be\label{casimir}
\{I^n_\gamma,\cdot\}_\gamma = 0.
\ee
Following the viewpoint of Remark \ref{frozen} it is natural to apply a kind of ``shift of argument'' approach to functions $\Trr\mu^n$, which are Casimirs for $\pb_0$. In this case, the argument would need to be shifted by the operator $P_+$, which, as mentioned before, does not belong to $\urp$. In effect the trace $\Trr$ would not make sense. The idea is to add extra terms to the Casimirs for the formula to be well defined. In \cite{GO-grass}, a family of Casimirs was found following this general idea (see formula~(3.7) there). Originally, it was written for the complex Banach Lie--Poisson space $\glrp := \urp \oplus i \urp$ but we present here the expression adapted to the case of the real Banach Lie--Poisson space $\urp$
\be \label{cas}I^n_\gamma(\mu):={\rm i}^{n+1}\Trr\bigl( (\mu-\gamma P_+)^{n+1}-(-\gamma)^{n}(\mu-\gamma P_+)\bigr).\ee
The proof that those are indeed Casimirs goes through the application of the coadjoint representation \eqref{coad-ext} of the central extension $\Urt$, see \cite[equation~(3.2)]{GO-grass}.

We will use the following fact from the Magri method approach (see, e.g., \cite{bols-bor,laurent-gengoux,magri} or~\cite[Appendix~B]{GO-grass}).
\begin{prop}
The coefficients of the expansion of the Casimirs of the Poisson bracket $\pb_\gamma$ in terms of the pencil parameter $\gamma$ are in involution with respect to all the brackets from this pencil. In consequence, the flows of the associated Hamiltonian vector fields commute and are contained in the intersection of the symplectic leaves of those brackets.
\end{prop}

Here, for simplicity, we expand the Casimirs into the series with respect to the parameter~$-\gamma$ instead. We denote by $W^n_k$ the polynomials defined by
$(\mu+\gamma P_+)^n=\sum_{k=0}^n\gamma^k W_k^n(\mu)$.
They are polynomials of degree $n-k$ in $\mu$ and contain up to $k$ occurrences of the projector $P_+$.
This leads to the expansion of Casimirs in the form
\be \label{W-expansion}I^n_\gamma(\mu)={\rm i}^{n+1}\sum_{k=0}^{n}(-\gamma)^k\Trr W_k^{n+1}(\mu)+{\rm i}^{n+1}(-\gamma)^{n}\Trr\mu.\ee
In this way, we obtain a family of functions on $\urp$
\be\label{calki}
h_k^n(\mu)={\rm i}^{n+1}\Trr W_k^{n+1}(\mu),\qquad 0\leq k\leq n.\ee
Plugging in the formula \eqref{calki} into the condition \eqref{casimir} gives rise to a bi-Hamiltonian hierarchy of equations of motion
\be \label{magri-chain}
\{h_k^n,\cdot\}_0=\{h_{k+1}^n,\cdot\}_s\ee
known also as a Magri chain.
Iterating the relationship \eqref{magri-chain}, we conclude that functions $h_k^n$
are in involution on $\urp$ for all Poisson brackets
$ \{h^n_k,h^m_l\}_0=\{h^n_k,h^m_l\}_s=0$
and thus can be considered as a family of Hamiltonians for an integrable system.

Note that we discarded the last term of \eqref{W-expansion} as it is a Casimir for $\pb_0$ and $\pb_s$. Hence, does not lead to any dynamics.

The Hamilton equations for the functions $h_k^n$ with respect to the Lie--Poisson bracket $\pb_0$ assume the following Lax form:
\be \label{W-eq}\frac{\partial}{\partial\tau_k^n} \mu = - {\rm i}^{n+1}(n+1) [\mu, W_k^n(\mu)] \ee
or, equivalently,
\be \label{W-eq-P}\frac{\partial}{\partial\tau_k^n} \mu = {\rm i}^{n+1}(n+1) [P_+, W_{k-1}^n(\mu)],\ee
where we have used the bi-Hamiltonian condition \eqref{magri-chain} and where $\tau_k^n$ denotes the time variable corresponding to the Hamiltonian $h^n_k$.
However, in equations \eqref{W-eq} and \eqref{W-eq-P} the functions $W^n_k$ are linearly dependent. For instance, one can observed that \[
W^k_{k-1}(\mu) = \mu P_+ + P_+ \mu + (k-2)P_+\mu P_+
\] for $k\geq 2$. Consequently, the family $\bigl\{W^k_{k-1}(\mu)\bigr\}_{k\geq 2}$ spans only a two-dimensional subspace generated by $\mu P_+ + P_+\mu$ and $P_+\mu P_+$. Thus, one might modify the system in question by introducing linearly independent homogeneous polynomials
\be \label{Hnk} H_k^n(\mu):= \sum_{
\genfrac{}{}{0pt}1{i_0,i_1,\dots, i_n\in\{0,1\}}{i_0+\dots +i_n=k}
} P_+^{i_0}\mu P_+^{i_1}\mu\cdots \mu P_+^{i_n}\ee
of degree $n\in\N$ in the operator variable $\mu\in \urp$ and degree $k$ in the operator $P_+$, where $k\leq n+1$.
In \cite[Proposition~3.1]{GO-grass}, it was shown that one can express the polynomials $\{W^n_k\}_{0\leq k\leq n\in\N } $ as linear combinations of $\{H_k^n\}_{0\leq k\leq n+1\in\N }$ and vice-versa.
Now consider the following hierarchy of Lax equations:
\be \label{H-eq}\frac\partial{\partial t_k^n}\mu = {\rm i}^{n+1} [\mu, H_k^n(\mu)],\ee
where $n\in\N$ and $k=1,\dots, n+1$ and the time variables $t^n_k$ are linear combinations of the time variables $\tau^n_k$. The corresponding flows commute since the right-hand side of \eqref{H-eq} are vector fields on $\urp$ that can be expressed as linear combinations of the commuting Hamiltonian vector fields in \eqref{W-eq} corresponding to the functions $h^n_k$. These flows are still bi-Hamiltonian with respect to $\pb_0$ and $\pb_s$ (since the space of Hamiltonian vector fields is a vector space). However, the explicit expression for the corresponding Hamiltonians is not easy to write.

Moreover, the polynomials $\{H_k^n\}_{0\leq k\leq n+1\in\N }$ satisfy the recurrence relation
\be \label{H-recc} H^{n+1}_{k+1}(\mu)=P_+\mu H_k^n(\mu)+\mu H^n_{k+1}(\mu)\ee
for $n\in \N$, where for consistency we use the convention $H^n_k=0$ for $k>n+1$ or $k<0$.
One can also write a ``dual'' recurrence relation by applying adjoint to \eqref{H-recc}
\be\label{H-recc-dual} H^{n+1}_{k+1}(\mu)=H_k^n(\mu)\mu P_+ +H^n_{k+1}(\mu)\mu \ee
for $n\in \N$.

\begin{prop}\label{prop:diag-const}
The diagonal blocks $\mu_{++}$ and $\mu_{--}$ are constants of motion
\smash{$\frac\partial{\partial t_k^n}\mu_{++} = 0$}, \smash{$\frac\partial{\partial t_k^n}\mu_{--} = 0$}.
\end{prop}
\begin{proof}
A commutator of an operator with $P_+$ has diagonal blocks equal to zero. Thus, from~\eqref{W-eq-P} it follows that the flows given by Hamiltonians $h^n_k$ keep the diagonal blocks constant.
Equations~\eqref{H-eq} are linear combinations of equations \eqref{W-eq-P}, so they retain this property.
\end{proof}

This fact has also two geometrical explanations, which we will present in the next section.

\section{Momentum map and symplectic leaves}\label{sec:momentum}

\subsection{Momentum map and Noether theorem}
First, let us notice that there is a natural coadjoint action of the Banach Lie group $\Ur$ on the Banach Lie--Poisson space $\urp$. We will write the momentum map in this case following the conventions from \cite[Definition 11.2.1]{ratiu-ms}. Since we are dealing with non-reflexive Banach spaces here, we replace the dual of the Lie algebra with a predual according to \cite[Definition~8.1]{OR}.

More precisely, $\urp$ sits inside the dual space to the Banach Lie algebra of $\Ur$, and is stable by coadjoint action. The left coadjoint action $g\cdot \mu$ of $g\in \Ur$ on $\mu\in \urp$ is defined by duality as
\[
\langle g\cdot \mu , \xi \rangle := \bigl\langle \mu , \Ad_g^{-1}(\xi)\bigr\rangle = \Trr \mu g^{-1}\xi g = \Trr g \mu g^{-1} \xi = \bigl\langle g\mu g^{-1} , \xi \bigr\rangle,
\]
where $\mu \in \urp$, $\xi\in \ur$ and $g\in \Ur$, where we have used the property \eqref{trr-cyclic}. The non-degeneracy of the duality pairing \eqref{pair-urp} then implies
$
g\cdot \mu = g \mu g^{-1} = \Ad^*_g(\mu)$.
The infinitesimal generator \smash{$X^\xi$} for the infinitesimal coadjoint action of $\xi \in \ur$ is defined as
\be\label{infinitesimal} X^\xi(\mu) = - \ad^*_\xi(\mu) = -\Trr \mu [\xi, \cdot],\ee
where $\mu\in\urp$. Using \eqref{trr-cyclic} and the duality pairing \eqref{pair-urp}, it follows that
$
 X^\xi(\mu) = [\xi,\mu]\in \urp$.
\begin{defn}
A momentum map for the action of a Lie group $G$, with Lie algebra $\g$, on a~Poisson manifold $M$ is a function $J\colon M\rightarrow \g^*$ such that, for $\mu\in M$ and $\xi \in \g$,
the function~${\mu\mapsto j_\xi(\mu) := \langle J(\mu) , \xi \rangle}$ is a Hamiltonian for the infinitesimal action of $\xi$ on $M$, i.e.,
\[
\{j_\xi, f\}(\mu) = X^\xi(f)(\mu) = \bigl\langle Df(\mu), X^\xi(\mu)\bigr\rangle.
\]
\end{defn}
In our case, the momentum map takes naturally values in the predual $\uujp\oplus\uujm$ of the Lie algebra $\uup\oplus\uum$:
\begin{prop}
The momentum map $J\colon\urp \to \uujp\oplus \uujm$ for the action of the subgroup $\Up\times\Um\subset \Ur$ on the Poisson manifold $\bigl(\urp, \pb_0\bigr)$ is
$J(\mu) = p_D(\mu)$,
where $p_D$ is the projection onto block-diagonal part
\[
 p_D(\mu) = P_+ \mu P_+ + P_- \mu P_- \in\uujp\oplus\uujm.\]
Moreover, it is equivariant with respect to the action of the group $\Up\times\Um$.
\end{prop}
\begin{proof}
Note that it is just the usual momentum map for the action of the Lie subalgebra $\mathfrak h\subset \g$ on the Lie--Poisson space $\g_*$ (or $\g^*$ in finite-dimensional context), see, e.g., \cite[Example~11.4.g]{ratiu-ms}.

Let us however go through the calculations to better illustrate the current context.
Comparing~\eqref{infinitesimal} with \eqref{blp-hvf}, we see that the vector field $X^\xi$ generated by the infinitesimal coadjoint action of $\xi\in \uup\oplus\uum$ is the Hamiltonian vector field corresponding to a function $j_\xi$ whose differential $Dj_\xi(\mu)$ at $\mu$ is simply $\xi$. One sets the integration constant to zero in order to obtain a linear map
$ j_\xi(\mu) = \pair\mu\xi$.
Now, since $\xi$ belongs to the Lie subalgebra $\uup\oplus\uum$, we would like to write $j_\xi(\mu)$ as a pairing with an element in $\uujp\oplus\uujm$. Using the block decomposition \eqref{blocks}, one sees that
$ j_\xi(\mu) = \pair{p_D(\mu)}\xi$,
where $p_D(\mu)\in \uujp\oplus\uujm$. It follows that the momentum map $J\colon \urp\rightarrow \uujp\oplus\uujm$ for the infinitesimal coadjoint action of $\uup\oplus\uum$ is
$ J(\mu) = p_D(\mu)$.
Equivariance means
$J(g\cdot\mu) = \Ad^*_g J(\mu)$
and is now straightforward.
\end{proof}

\begin{lem}
The family of Casimirs \eqref{cas} is preserved by the action of the subgroup $\Up\times\Um\subset \Ur$.
\end{lem}
\begin{proof}
Since the subgroup $\Up\times\Um\subset \Ur$ is the stabilizer of $P_+$, for any element~${u\in\Up\times\Um}$, we have $u^* P_+ u = P_+$. Using the formula \eqref{cas}, we obtain
\begin{align*} I^n_\gamma(u^* \mu u):={}&{\rm i}^{n+1}\Trr\bigl( (u^*\mu u-\gamma P_+)^{n+1}-(-\gamma)^{n}(u^*\mu u-\gamma P_+)\bigr)\\
={}&
 {\rm i}^{n+1}\Trr u^*\bigl( (\mu-\gamma P_+)^{n+1}-(-\gamma)^{n}(\mu-\gamma P_+)\bigr)u.\end{align*}
We get the claim using \eqref{trr-cyclic}.
\end{proof}

\begin{lem}
The Hamiltonians \eqref{calki} are also invariant with respect to the action of the subgroup~${\Up\times\Um\subset \Ur}$.
\end{lem}
\begin{proof}
Since we have
\[ (u^* \mu u+\gamma P_+)^n=u^*(\mu +\gamma P_+)^n u=u^*\Biggl(\sum_{k=0}^n\gamma^k W_k^n(\mu)\Biggr)u, \]
we conclude that
$W^n_k(u^*\mu u) = u^* W^n_k(\mu) u$.
In effect $h^n_k(u^*\mu u) = h^n_k(\mu)$ by application of formula~\eqref{trr-cyclic}.
\end{proof}

\begin{cor}
By two previous lemmas, it follows by means of the Noether theorem $($see, e.g., {\rm\cite[\emph{Theorem}~8.2]{OR})} that the level sets $J^{-1}(c)$, $c\in \uujp\oplus \uujm$, are conserved by the Hamiltonian flows, which is equivalent to the observation that $\mu_{++}$ and $\mu_{--}$ are constant.
\end{cor}

\begin{rem}
While the Casimirs \eqref{cas} are preserved by the coadjoint action of the central extension $\Urt$ of $\Ur$, the Hamiltonians are not. It is a consequence of the fact that this action maps $\mu$ to a first order polynomial in $\gamma$, see \eqref{coad-ext}.
Hence, it mixes the terms in the series expansion \eqref{W-expansion}.
\end{rem}

\subsection{Magri method and symplectic leaves}
The second observation is related to the symplectic leaves.
It is a consequence of the bi-Hamiltonian structure \eqref{magri-chain} and Magri method that the Hamiltonian flows corresponding to the family of functions $\{h^n_k\}$ defined in \eqref{calki} preserve the symplectic leaves of both Poisson structures appearing in the Poisson pencil.
In our case, it means that both the symplectic leaves of the Lie--Poisson bracket $\pb_0$ and the one given by Schwinger term $\pb_s$ \eqref{pb-s}
are preserved.

Using equation~\eqref{schwinger_expression}, it is easy to see that the kernel of the bilinear mapping $s$ contains the space of block-diagonal operators which are exactly those commuting with $P_+$. Moreover, when restricted to the complement of block-diagonal operators given by off-diagonal operators, the Schwinger term is non-degenerate. In particular, when restricted to the subspace
\[
\mathfrak{m} =
\left\{\left(
\begin{matrix}
 0 & \mu_{+-}\\
 -(\mu_{+-})^* & 0\\
\end{matrix}
\right) \mid \mu_{+-}\in\Ldpm\right\}\subset \urp,
\]
it defines a symplectic bilinear form
\[
s(\mu,\nu)=\Tr(-\mu_{+-}\nu^*_{+-}+\nu_{+-}\mu^*_{+-}) = 2{\rm i} \operatorname{Im} \Tr \nu_{+-}\mu^*_{+-},
\]
which is used to construct the symplectic form on $\Gr$, see, e.g., \cite{Ratiu-grass,wurzbacher}. More precisely, the symplectic form on the Banach space $\mathfrak{m}\cong T_{\H_+}\Gr$ is propagated to the whole tangent bundle by translation with the group $\Ur$ (since $\Gr$ is a homogeneous space with respect to this group).
Since the Poisson bracket $\pb_s$ is constant, it follows that the symplectic leaves for $\pb_s$ are the affine subspaces of $\urp$ defined by %obtained by shifting the Banach spaces
\[
\left(
\begin{matrix}
 A & 0\\
 0 & D
\end{matrix}
\right) + \mathfrak m \subset \urp \]
for fixed operators $A\in\uujp$ and $D\in\uujm$.
In consequence, we again conclude that diagonal blocks of $\mu$ are constant with respect to the flows given by \eqref{W-eq}. For more discussions about symplectic leaves in Banach Lie--Poisson spaces, see \cite{beltita05}.

\section{Equations on the groupoid of partial isometries}\label{sec:partiso}

In this section, we will restrict our attention to a special case of equations \eqref{H-eq} given by the condition $\mu_{++}=0$, which give rise to the evolution on the set of partial isometries.

Let us recall briefly here that by a partial isometry we mean an operator $u$ acting on the Hilbert space $\H$ such that it restricts to a unitary map between $(\ker u)^\perp$ and $\im u$. The space~$(\ker u)^\perp$ is called the initial space of partial isometry $u$ and will be denoted $s(u)$, while the space $\im u$ is called the final space and will be denoted $t(u)$. Alternatively, partial isometries can be characterized by any of the conditions
$u^*uu^*=u^*$, $uu^*u=u$, $ u^*u$ is a projection (on $ s(u)$), $uu^*$ is a projection (on $t(u)$).
The set of all partial isometries of the Hilbert space has a natural structure of Banach Lie groupoid, see \cite{OS}. In the sequel, we will use also partial isometries between two different Hilbert spaces.

Let us begin by presenting an easier situation --- namely we will restrict our attention to the hierarchy of equations \eqref{H-eq} with $k=1$.

\subsection[Special case of the hierarchy of equations with k=1]{Special case of the hierarchy of equations with $\boldsymbol{k=1}$}

In this situation, the operators \eqref{Hnk} contain only one occurrence of $P_+$, and the hierarchy of equations \eqref{H-eq} assume the explicit form
\[
 \frac\partial{\partial t^n_1} \mu = {\rm i}^{n+1}\bigl[\mu, P_+\mu^n + \mu P_+ \mu^{n-1}+\dots + \mu^n P_+\bigr].\]
After expanding the commutator and canceling out most of the terms we arrive at
\be\label{eqk1} \frac\partial{\partial t^n_1} \mu = -{\rm i}^{n+1}\bigl[P_+,\mu^{n+1}\bigr]\ee
or in the block decomposition
\bse\label{h0}%\label{h0b}
\begin{empheq}{align}
&\frac\partial{\partial t^n_1} \mu_{++}=0,\label{h0a}\\
&\frac\partial{\partial t^n_1} \mu_{+-}=-{\rm i}^{n+1}\bigl(\mu^{n+1}\bigr)_{+-},\qquad\frac\partial{\partial t^n_1} \mu_{-+}={\rm i}^{n+1}\bigl(\mu^{n+1}\bigr)_{-+},\label{h0c}\\
&\frac\partial{\partial t^n_1} \mu_{--}=0\label{h0d},
\end{empheq}
%\right.
\ese
where equations \eqref{h0a} and \eqref{h0d} are also consequences of Proposition~\ref{prop:diag-const}.
Following the idea from \cite{G-dr}, we get the following proposition.

\begin{prop}\label{modulus}
In the case $\mu_{++}=0$, the modulus \smash{$\abs{\mu_{-+}}= \sqrt{\mu_{-+}^*\mu_{-+}}$} is constant with respect to times $t^n_1$.
\end{prop}

\begin{proof}
One can straightforwardly compute that
\[ \frac\partial{\partial t^n_1} (\mu_{+-}\mu_{-+}) = {\rm i}^{n+1}\bigl[\bigl(\mu^{n+1}\bigr)_{++},\mu_{++}\bigr].\]
Now if we assume that the block $\mu_{++}=0$ for all $t^n_1$, we see that $\abs{\mu_{-+}}$ is constant.
\end{proof}

\begin{cor}\label{cor_polar}
The hierarchy of equations \eqref{eqk1} induces a family of commuting flows on the set of partial isometries $u$
with fixed initial space $s(u)$,
depending on
a fixed positive operator~${B\in L^1_{++}}$ such that $\overline{\im B} = s(u)$ and a fixed skew adjoint operator $D \in L^1_{--}$.
\end{cor}

\begin{proof}

Consider $\mu$ given in the following form:
\be\label{mu-partiso}
\mu = \left(
\begin{matrix}
 0 & -Bu^*\\
 uB & D
\end{matrix}
\right).\ee
In this formula, $\mu_{-+}$ is given by the polar decomposition
\be\label{polar} \mu_{-+} = u B,\ee
where $u\colon\H_+\to\H_-$ is the unique partial isometry satisfying \eqref{polar} with initial space $s(u)=\overline{\im B}$.

By Proposition~\ref{prop:diag-const}, $\mu_{++}=0$ and $\mu_{--}=D$ are constants of motion for all the flows in the hierarchy. By Proposition~\ref{modulus}, if additionally $\mu_{++} = 0$, the modulus $\abs{\mu_{-+}} = B$ is also constant. Thus, the form of the operator $\mu$ is preserved. By uniqueness of the polar decomposition, each flow on $\mu_{-+}$ given in \eqref{h0c} induces a flow on $u$, which fixes the initial space $s(u)$.
\end{proof}

\begin{rem}
Observe that in general the dependence with respect to time of the partial isometry $u$ coming from the polar decomposition needs not to be smooth (or even continuous). It can be demonstrated even in the simple case of a curve $\phi(t)=M t$ for a non-zero matrix $M$. Obviously, at the point $t=0$ the partial isometry from the polar decomposition of $\phi(t)$ is zero, while it is a non-zero operator for other values of $t$.
\end{rem}

\begin{prop}\label{prop-part-iso1}
The equations on the partial isometry $u$ induced by the hierarchy of equations~\eqref{h0c} can be written as follows:
\be \label{eq-part-iso1}\frac\partial{\partial t^n_1} u = {\rm i}^{n+1} (\mu^{n})_{--} u\ee
for $\mu$ given by \eqref{mu-partiso}.
\end{prop}

\begin{proof}
From the condition $\mu_{++} = 0$ it follows that $\bigl(\mu^{n+1}\bigr)_{-+} = (\mu^n)_{--}\mu_{-+}$. Thus, the equation~\eqref{h0c} assumes the form
\[
\frac\partial{\partial t^n_1} uB = {\rm i}^{n+1} (\mu^{n})_{--} uB.
\]
Thus, equation \eqref{eq-part-iso1} is satisfied on $\im B$ and by continuity of $u$, also on \smash{$\overline{\im B}$}. By definition, $u$~is zero on \smash{$\overline{\im B}^\perp$}, what implies the stated result.
\end{proof}

\begin{rem}
 The partial isometry $u$ can be extended trivially to a partial isometry in $\H$. In this way, we obtain a family of differential equations on the Banach Lie groupoid of partial isometries $\UU$ constructed in \cite{OS} generating a flow on
$s^{-1}\bigl((\ker B)^\perp\bigr)\cap t^{-1}(\operatorname{Gr}(\H_-))\subset \UU$, where $\operatorname{Gr}(\H_-)$ is the Grassmannian of all closed subspaces of $\H_-$.
\end{rem}
For values $n=1$ and $n=2$, the flow \eqref{eq-part-iso1} is linear
\begin{gather*} \frac\partial{\partial t^1_1} u = -Du,\qquad
\frac\partial{\partial t^2_1} u = {\rm i}\bigl(uB^2-D^2u\bigr),
\end{gather*}
while for $k=3$ it assumes the following form
\smash{$\frac\partial{\partial t^3_1} u = -DuB^2-uB^2u^*Du+D^3u$},
which taken together with its adjoint can be seen as a pair of coupled Riccati equations on $u$ and $u^*$.

The projector $u^*u$ onto the initial space $s(u)$ is naturally constant by construction, but the projector $uu^*$ onto the final space $t(u)$ satisfies the following equation in Lax form
\[
\frac\partial{\partial t^n_1} (uu^*) = -{\rm i}^{n+1} [uu^*,(\mu^{n})_{--}],
\]
which can be obtained directly from \eqref{eq-part-iso1}.
Note that the right-hand side depends on $u$ and $u^*$, so this equation cannot be viewed as an independent equation on the Grassmannian of~$\H$.

\subsection[The hierarchy of equations in the general case k<= n+1]{The hierarchy of equations in the general case $\boldsymbol{k\leq n+1}$}

Now let us return to the study of the general case for an arbitrary value of $k$. First of all, we will need a formula that comes from applying the recurrence \eqref{H-recc} twice
\be \label{H-recc2}H^n_k(\mu)= P_+\mu P_+ \mu H^{n-2}_{k-2}(\mu) +
\bigl(P_+\mu^2 + \mu P_+ \mu \bigr) H^{n-2}_{k-1}(\mu) + \mu^2 H^{n-2}_{k}(\mu).\ee
Using this formula, we will be in position to prove the following fact.

\begin{prop}
In the case $\mu_{++}=0$, the modulus $\abs{\mu_{-+}}$ is constant along the bi-Hamiltonian flows for all $t^n_k$, $n\in \N$, $k\leq n+1$.
\end{prop}

\begin{proof}
Just like previously, we compute
the time derivative of $\mu_{-+}^*\mu_{-+}$ but we use the equation in the general form \eqref{H-eq}. After some work, we obtain
\[
 \frac\partial{\partial t^n_k} (\mu_{-+}^*\mu_{-+}) = {\rm i}^{n+1}P_+\bigl(\bigl[\mu^2,H_k^n(\mu)\bigr] +
\mu P_+ [H^n_k(\mu),\mu] +
[H^n_k(\mu),\mu] P_+ \mu\bigr) P_+.
\]
Note that all terms but the first one vanish since $\mu_{++}=P_+\mu P_+=0$. Now we need to prove that the first term vanishes as well. To this end, let us apply formula \eqref{H-recc2} to one term of the commutator and the adjoint formula
to the other. We assume the convention that $H_k^n(\mu)=0$ for $k>n+1$, so we can avoid worrying about the range of indices $k$ and $n$. In this way, we obtain after some cancellations
\begin{align*}
 P_+\bigl[\mu^2,H_k^n(\mu)\bigr]P_+ ={}&
P_+ \mu^2 \bigl(H^{n-2}_{k-2}(\mu) + H^{n-2}_{k-1}(\mu)\bigr) \mu P_+ \mu P_+\\
& -
P_+ \mu P_+ \mu \bigl(H^{n-2}_{k-2}(\mu) + H^{n-2}_{k-1}(\mu)\bigr) \mu^2 P_+ .
\end{align*}
Again, we conclude that this expression vanishes if $\mu_{++}=0$.
\end{proof}

As in Corollary \ref{cor_polar}, when using the polar decomposition of $\mu_{-+}=uB$, the hierarchy of equations \eqref{H-eq} gives rise to a family of commuting equations on the partial isometry $u$, which can be written down explicitly.

\begin{prop}
Assume that $\mu_{++} = 0$.
The equations for the evolution of the partial isometry~$u$ assume the form
\be \label{eq-partiso} \frac\partial{\partial t^n_k} u = {\rm i}^{n+1}\bigl(\mu H^{n-1}_{k-1}(\mu)\bigr)_{--}u\ee
for $n\in\N$, $k\leq n+1$.
\end{prop}
\begin{proof}
Directly from equation \eqref{H-eq} using the condition $P_+ \mu P_+ = 0$, we get
\[
 \frac\partial{\partial t^n_k} uB = {\rm i}^{n+1}P_-[\mu, H^n_k(\mu)]P_+ =
{\rm i}^{n+1}(P_-\mu H^n_k(\mu) P_+ - P_- H^n_k(\mu) P_- \mu P_+).\]
Now, we use the rule \eqref{H-recc-dual} in the first term and the rule \eqref{H-recc} in the second and obtain
\[
 \frac\partial{\partial t^n_k} uB = {\rm i}^{n+1}P_-\bigl(\mu H^{n-1}_{k-1}(\mu) + \mu H^{n-1}_{k}(\mu)
-P_+\mu H^{n-1}_{k-1}(\mu) - \mu H^{n-1}_k(\mu)\bigr)uB.\]
After simplifying the expression, we obtain expression \eqref{eq-partiso} multiplied on the right by the operator $B$. The same argument as in the proof of Proposition \ref{prop-part-iso1} gives the stated result.
\end{proof}

\begin{prop}\label{rem:limitk}
Since $\mu_{++}=0$, the right-hand side of equation \eqref{eq-partiso} vanishes for $k> n/2+1$.
\end{prop}

\begin{proof}
Recall from formula \eqref{Hnk} that $H^{n-1}_{k-1}$ is a sum of expressions of the form
$ P_+^{i_0}\mu P_+^{i_1}\mu\cdots \allowbreak\smash{\mu P_+^{i_{n-1}}}$
with $i_0, i_1, \dots, i_{n-1}\in\{0,1\}$ and $i_0+i_1+\dots+i_{n-1} = k-1$.
The expressions ending with a $P_+$ will not contribute in formula \eqref{eq-partiso}, since
$
\bigl(\mu P_+^{i_0}\mu P_+^{i_1}\mu\cdots \mu P_+\bigr)_{--} = 0$,
because $P_+ P_- = 0$.
Moreover, if $k> n/2+1$ all expressions of the form
\smash{$P_+^{i_0}\mu P_+^{i_1}\mu\cdots \mu P_+^{i_{n-2}}\mu$}
vanish, since necessarily at least two subsequent indices $i_j$, $i_{j+1}$ need to be non-zero and \smash{$P_+^{i_j}\mu P_+^{i_{j+1}} = \mu_{++} = 0$},
see also~\cite[Proposition~2]{GT-partiso-bial}.
\end{proof}

\subsection{Example: rank one case}\label{ex:rank1}

Following the ideas from the finite-dimensional case described in \cite{GT-partiso-bial}, we will now present an example of solutions in the case when the rank of the partial isometry $u$ is equal to 1. The approach we use is a generalization of the methods presented in \cite[Section 4]{GT-partiso-bial}. From the properties of polar decomposition (see Corollary \ref{cor_polar}), it follows that also $\dim\im B=1$.
Since~$B$ and ${\rm i}D$ are self adjoint trace-class and thus compact, one can choose an orthonormal basis~$\set{e_i}$ in $\H_+$ and $\set{f_i}$ in $\H_-$ in which they are diagonal.
We can further choose the basis in such a~way that $\im B$ is spanned by the first basis vector $e_1$, so the operator $B$ can be written as~${ B = b\proj{e_1}{e_1}}$
for some $\R\ni b>0$. In consequence, the partial isometry $u$ becomes an operator of the form
$u = \proj\psi{e_1}$
for some
$\psi = \sum_{j=1}^\infty \alpha_j f_j\in\H_-$
with norm 1, see also formula (28) in~\cite{GT-partiso-bial} for matrix expression of $u$ in finite dimensions.

In consequence, we obtain a hierarchy of equations on the single vector $\psi$ on the unit sphere of $\H_-$
\be \label{eq-partiso2} \frac\partial{\partial t^n_k} \psi = {\rm i}^{n+1}\bigl(\mu H^{n-1}_{k-1}(\mu)\bigr)_{--}\psi\ee
depending on the single non-zero eigenvalue $b$ of $B$ and (possibly infinitely many) eigenvalues~${d_1, d_2,\ldots\in {\rm i}\R}$ of $D$

\begin{prop}\label{prop_rank1}
The equations for the evolution of the coefficients $\alpha_1,\alpha_2,\dots$ of the vector $\psi$ with respect to the arbitrary time $t^n_k$ are of the form
\begin{gather*}%\label{eq-alpha}
 \frac\partial{\partial t^n_k} \alpha_j = {\rm i} p_{j,k}^n\bigl(\abs{\alpha_1}^2,\abs{\alpha_2}^2,\dots\bigr) \alpha_j,
\end{gather*}
where $p_{j,k}^n$ are smooth real-valued functions depending on the eigenvalues of the matrices $B$ and~$D$. Moreover, for fixed $n$, $k$, the absolute values of $p_{j,k}^n$ are bounded with respect to $j$ by some sequence of positive constants $\{w^n_k\}_{n\geq 0, k\geq 0}$:
\[\bigl|p_{j,k}^n\bigl(\abs{\alpha_1}^2,\abs{\alpha_2}^2,\dots\bigr)\bigr|\leq w^n_k.\]
\end{prop}
\begin{proof}
We follow the steps of the proof of \cite[Proposition 3]{GT-partiso-bial}.
First, we note that $Bu^*u = B$.
We see that the right-hand side of \eqref{eq-partiso2}
is a sum of terms of the form
\begin{gather}
{\rm i}^{n+1} \bigl(D^{i_1} u B^{j_1} u^*\bigr)\cdot \bigl(D^{i_2} u B^{j_2} u^*\bigr) \cdots \bigl(D^{i_l} u B^{j_l} u^*\bigr) u\nonumber
\\
\qquad= {\rm i}^{n+1} D^{i_1} u B^{j_1} \bigl(u^* D^{i_2} u B^{j_2}\bigr) \cdot \bigl(u^* D^{i_3} u B^{j_3}\bigr)\cdots \bigl(u^*D^{i_l} u B^{j_l}\bigr) u^* u\label{gen_term}
\end{gather}
for some $i_1,\dots, i_l, j_1,\dots, j_l\in \{0,1,\dots\}$.
In the second line of this formula, we moved the parentheses in order to get products of terms of the form $u^*D^s u B^r$. Then we observe that
\[
 u^*D^s u B^r = b^r \sc{\psi}{D^s\psi} \proj{e_1}{e_1} = b^r \sum_{i=1}^\infty d_i^s \abs{\alpha_i}^2 \proj{e_1}{e_1} .\]
Note that $d^s_i$ are bounded by $\norm{D^s}$ and the infinite sum above converges in the norm topology. In the end plugging this equality into \eqref{gen_term} and summing up, we get the equation of the form
\be\label{eq_psi}
 \frac\partial{\partial t^n_k} \psi = {\rm i}^{n+1} \sum_i b^{k_i} \Biggl(\prod_{j=0}^{r_i} \sc{\psi}{D^{s_{ij}}\psi}\Biggr) D^{l_i}\psi
\ee
for some $k_i,l_i,r_i,s_{ij}\in\N$. The expression on the right-hand side is a diagonal operator, whose eigenvalues are denoted by $p_{j,k}^n$, acting on the vector $\psi$.
Splitting this equation into components yields the thesis.
The functions $p_{j,k}^n$ are essentially algebraic combinations of expressions~$\sc{\psi}{D^s\psi}$, which smoothly depend on \smash{$\abs{\alpha_1}^2, \abs{\alpha_2}^2, \dots$}. Moreover, for fixed $n$ and $k$ they are bounded by the norm of the operator on the right-hand side of \eqref{eq_psi}. This norm in general depends on the values of $\abs{\alpha_1}^2, \abs{\alpha_2}^2, \dots$, but by means of triangle inequality and the fact that~${\abs{\alpha_j}\leq 1}$, we get the required bound.

It remains to show that the functions $p_{j,k}^n$ are real. It follows from the fact that the matrix in front of the matrix $u$ on the right-hand side of equation \eqref{eq-partiso} is skew hermitian.
\end{proof}

\begin{thm}\label{thm-rank1}
The solution to \eqref{eq-partiso2} for the case of partial isometries of rank one is the following:
\[\psi\bigl(t_1^1,t_1^2,t_2^2,\dots\bigr) = \sum_{j=1}^\infty \alpha_j\bigl(t_1^1,t_1^2,t_2^2,\dots\bigr) f_j\]
for
\begin{equation}\label{solution}
\alpha_j\bigl(t_1^1,t_1^2,t_2^2,\dots\bigr) = \alpha_j^0 \exp\Biggl({{\rm i} \sum_{n,k\leq n/2+1} p_{j,k}^n\bigl(\abs{\alpha_1^0}^2,\abs{\alpha_2^0}^2,\dots\bigr)t^n_k}\Biggr),
\end{equation}
where $\alpha_j^0\in\C$ is the initial value of the $j$-th component of vector $\psi$. In order to ensure convergence, we assume that the sequence of times $t^n_k$ satisfies the condition
\be \label{times_cond} \sum_{n,k\leq n/2+1} w^n_k \abs{t^n_k}^2 < \infty. \ee
\end{thm}
\begin{proof}We consider only $k\leq n/2+1$ since due to Proposition~\ref{rem:limitk} all other flows are trivial.
Using the polar form of the coefficients
$\alpha_j = r_j {\rm e}^{{\rm i}\varphi_j}$
we obtain from Proposition~\ref{prop_rank1} the equations in the following form:
\[
 \frac\partial{\partial t^n_k} r_j = 0 ,\qquad
 \frac\partial{\partial t^n_k} \varphi_j = p_{j,k}^n\bigl((r_1)^2,(r_2)^2,\dots\bigr).
\]
In effect we conclude that
$r_j\bigl(t_1^1,t_1^2,t_2^2,\dots\bigr) = r_j(0, 0, 0, \dots) = \bigl|\alpha_j^0\bigr|$
and
\[\varphi_j\bigl(t_1^1,t_1^2,t_2^2,\dots\bigr)
= \varphi_j(0, 0, 0, \dots) + \sum\limits_{{n,k\leq n/2+1}} p_{j,k}^n\bigl(|\alpha_1^0|^2,|\alpha_2^0|^2,\dots\bigr)t^n_k.\]
Convergence of this series is ensured by condition \eqref{times_cond}.
Combining those results, we obtain~\eqref{solution}.

Note that since the absolute values of $\alpha_j$ are preserved, the norm of the vector $\psi$ is constant.\looseness=-1
\end{proof}

The solution can be alternatively expressed in terms of the operator $u$ as follows. Consider a sequence of diagonal operators in $\H_-$ with eigenvalues
\[
R_{n,k}=\operatorname{diag}\bigl(
p_{j,k}^n\bigl(\abs{\alpha_1^0}^2,\dots\bigr),\, j=1,2,\dots\bigr)
\]
for $n\in\N$, $k\leq n/2+1$.
Now the solution in terms of $u$ looks like this
\begin{align*}
u\bigl(t_1^1,t_1^2,t_2^2,\dots\bigr)&= \exp\bigl({\rm i}R_{1,1}{t_1^1}+
{\rm i}R_{1,2}{t_1^2}+{\rm i}R_{2,2}{t_2^2}+\cdots\bigr) u(0,0,\dots)\\
& =\exp\Biggl({{\rm i}\sum\limits_{{n,k\leq n/2+1}} R_{n,k}t^n_k}\Biggr)u(0,0,\dots). \end{align*}
Note that the operators $R_{n,k}$ depend on the initial value of $u$.

This approach does not work for partial isometries of higher rank. The solution obtained here was due to the fact that the operator $B$ acted effectively as a scalar what led to significant simplification of the equations. That will not be the case even for partial isometries of rank~2. In a more general case, another approach, possibly involving integrals of motion, would be needed.

\section[Flow on the restricted Grassmannian Gr]{Flow on the restricted Grassmannian $\boldsymbol{\Gr}$}\label{sec:gr-lin}

\subsection{Linearity of the flows on the restricted Grassmannian}
Let us consider $\Urt$, the central extension of $\Ur$ with Lie algebra $\urt$.
It was shown in \cite[Theorem 2.13]{Ratiu-grass} that the coadjoint orbits of $\Urt$ acting on $\urtp$ passing through the point $(0,\gamma)$, $\gamma\neq 0$, $\gamma\in {\rm i}\R$, are diffeomorphic to the restricted Grassmannian $\Gr$. Note that since the extension is central, the coadjoint action \eqref{coad-ext} does not change the second argument $\gamma$ and it gives rise to the so-called affine coadjoint action of $\ur$ on $\urp$. In effect, it is possible to view the orbit as a subset of $\urp$. The diffeomorphism is given by
\[
 \Phi_\gamma \colon\ \Gr \ni W \tto \mu = \gamma (P_W - P_+) \in \mathcal O_{(0,\gamma)}\subset \urp. \]
\begin{prop}
An element $\mu\in\urp$ belongs to the coadjoint orbit $\mathcal{O}_{(0,\gamma)}$ if and only if~$\frac{1}{\gamma}\mu + P_+$ is an orthogonal projection.
\end{prop}
\begin{proof}
The inverse of the map $\Phi_\gamma$ is
\be \label{phiGammaInverse}
 \Phi_\gamma^{-1}\colon\ \mathcal{O}_{(0,\gamma)} \ni \mu \tto \im(p) \in \Gr,\qquad\textrm{where}\quad p = \frac{1}{\gamma}\mu + P_+.
 \ee
Since $\Phi_\gamma$ is a diffeomorphism, $\mu\in\urp$ belongs to the coadjoint orbit $\mathcal{O}_{(0,\gamma)}$ if and only if~\smash{$\frac{1}{\gamma}\mu + P_+$} is a projection $P_W$ onto an element $W\in \Gr$.

Thus, what remains to be proven is that if $\frac{1}{\gamma}\mu + P_+$ is a projection it is necessarily a projection onto a subspace belonging to $\Gr$. That fact follows from Proposition \ref{prop:gr-proj} since~${\urp\subset \Ld}$.
\end{proof}

The coadjoint orbit is preserved by the flows \eqref{H-eq} and we can restrict the integrable system to it.
\begin{prop}\label{prop:mu-linear}
For initial conditions in the coadjoint orbit $\mathcal O_{(0,\gamma)}$, the equations \eqref{H-eq} are linear.
\end{prop}
\begin{proof}
By computing the square of the projection $p:= \frac{1}{\gamma}\mu + P_+$ and using $p^2 = p$ as well as the block decomposition of $\mu$, we get that
\be \label{mu2} \mu^2 = \gamma (\mu-\mu P_+ - P_+\mu) = \gamma(\mu_{--}-\mu_{++}).\ee
From this and from Proposition~\ref{prop:diag-const}, we conclude that $\mu^2$ is constant and block diagonal. Moreover, breaking it down into block form we conclude
\begin{gather}\label{mu2-block}
\mu_{++} \mu_{+-} = - \mu_{+-}\mu_{--},\qquad
\mu_{-+} \mu_{++} = - \mu_{--}\mu_{-+},\qquad
\mu_{-+}\mu_{+-} = \const,\nonumber\\
\mu_{+-}\mu_{-+} = \const.
\end{gather}

From \eqref{mu2-block}, it follows that the polynomials $H^n_k(\mu)$ are at most linear in $\mu_{-+}$ and $\mu_{+-}$. Namely, whenever we have an expression which contains an even number of those operators, it needs to be built out of terms like $\mu_{-+}\mu_{++}^l\mu_{+-}$ or $\mu_{+-}\mu_{--}^j\mu_{-+}$, which can be shown to be equal to \smash{$(-1)^l \mu_{--}^l\mu_{-+}\mu_{+-}$} or \smash{$\mu_{--}^j\mu_{-+}\mu_{+-}$} respectively. Hence, they are constant. Analogously, an expression with an odd number of those operators becomes linear in $\mu$ (up to the multiplication by constant).
\end{proof}

\subsection[Solutions of the hierarchy on the restricted Grassmannian for k = 1]{Solutions of the hierarchy on the restricted Grassmannian for $\boldsymbol{k = 1}$}
Let us now consider the open neighborhood $\Omega_{\H_+}$ of $\H_+$ consisting of elements in $\Gr$ such that the projection onto $\H_+$ is an isomorphism
\[
\Omega_{\H_+}= \{ W \in \Gr \mid {P_W}_{++}= P_+ P_W P_+ \textrm{ is invertible in } \H_+\}.
\]
A chart $\phi_{\H_+}$ on the restricted Grassmannian around $\H_+$ is
\be \label{phiH+}
\phi_{\H_+}\colon\ \Omega_{\H_+} \to \Ldpm,\qquad
 W \mapsto \phi_{\H_+}(W)= {P_W}_{-+}({P_W}_{++})^{-1},
\ee
and it inverse is
$\phi_{\H_+}^{-1}(A) = \Gamma(A)$,
where $\Gamma(A)$ is the graph of an operator. Moreover, the orthogonal projection on $W$ can be expressed as
\[ P_W = (1_{\H_+} + A)(1_{\H_+}+A^*A)^{-1}(P_+ + A^* P_-),\]
see, e.g., \cite{GJS-partiso}. Using \eqref{phiH+} and \eqref{phiGammaInverse}, one has, for $\mu \in \mathcal O_{(0,\gamma)}\cap \Phi_\gamma(\Omega_{\H_+})$,
\be \label{A} A = \phi_{\H_+}\circ\Phi_\gamma^{-1}(\mu) =
\mu_{-+}(\mu_{++}+\gamma P_+)^{-1}.\ee
Conversely, composing the chart \smash{$\phi_{\H_+}^{-1} $} with the diffeomorphism $\Phi_\gamma$ one obtains a parametrization of the restricted Grassmannian realized as a coadjoint orbit inside $\urp$
\be \label{coord}
\Phi_\gamma\circ\phi_{\H_+}^{-1}(A) = \gamma\left(
\begin{matrix}
 (1+A^*A)^{-1} - 1& (1+A^*A)^{-1}A^*\\
 A(1+A^*A)^{-1} & A(1+A^*A)^{-1}A^*
\end{matrix}
\right),
\ee
where $A\in \Ldpm$.

Let us now consider the equations \eqref{H-eq} in terms of an operator $A$ given as \eqref{A}. Note that relation \eqref{A} is actually linear in $\mu$ from the point of view of evolution according to the equations~\eqref{H-eq} since $\mu_{++}=\const$. Thus, from Proposition \ref{prop:mu-linear} we conclude that the equations~\eqref{H-eq} are also linear when expressed in the chart $\phi_{\H_+}$.

One can also notice that since $\abs A$ is constant, we again get the flow on the Banach Lie groupoid of partial isometries. However, it is not directly related to the case discussed in the previous section, where we assumed the condition $\mu_{++}=0$. Moreover, all equations on the partial isometry $u$ defined by the polar decomposition $A=u\abs A$ are linear.

Let us write down several of the equations in coordinates \eqref{coord}. First of all, from \eqref{mu2} it follows that all equations for $k=1$ and $n$ odd are trivial since $\mu^{2l}$ is diagonal
\smash{$\frac\partial{\partial t^{2l+1}_1} A = 0$}.
The equations for $k=1$ and $n$ even assume the form
\smash{$\frac\partial{\partial t^{2l}_1} A = (-1)^l {\rm i} R_l A$},
where
\[ R_l = \bigl(A (A^* A+1)^{-1} A^*\bigr)^l = \const. \]
In consequence, their solutions are
\[ A\bigl(t_1^1,t_1^2,t_1^3,\dots\bigr)= \exp\bigl(-{\rm i}R_{1}{t_1^2}+
{\rm i}R_{2}{t_1^4}-{\rm i}R_{3}{t_1^6}+\cdots\bigr) A(0,0,\dots) \]
with respect to the times $t_1^1,t_1^2,t_1^3,\dots$, where, in order to ensure convergence, we assume that the times $t^n_1$ decrease sufficiently rapidly, i.e.,
$\sum_{n\in\N} \norm{R_n} \bigl|t^{2n}_1\bigr|^2 \leq \infty$.

\subsection*{Acknowledgement}
This research of the authors was partially supported by joint National Science Centre, Poland (number 2020/01/Y/ST1/00123) and Fonds zur F\"orderung der wissenschaftlichen Forschung, Austria (number I 5015-N) grant ``Banach Poisson--Lie groups and integrable systems''.
The authors would like to thank the Erwin Schr\"odinger Institute for its hospitality during the thematic programme ``Geometry beyond Riemann: Curvature and Rigidity'' in September 2023. The stimulating atmosphere of the Workshop on Geometric Methods in Physics in Białowie\.za, Poland, is also acknowledged, where both authors could meet to work on this paper.
The authors are grateful to the referees for their remarks and suggestions, which improved the presentation of the results.

\pdfbookmark[1]{References}{ref}
\LastPageEnding

\end{document}